%% file: QMA-QCMA.tex
\documentclass[letterpaper,11pt]{article} 
\usepackage[letterpaper,hmargin=1in,vmargin=1in]{geometry} 
\usepackage{fullpage}
\usepackage[pdfstartview=FitH,colorlinks,linkcolor=blue,citecolor=blue]{hyperref} 

\usepackage[T1]{fontenc}
\usepackage{color} 
\usepackage{times} 
\usepackage{microtype,pdfsync} 
\usepackage{amsmath,amsfonts,amssymb}
\usepackage{amsthm} 
\usepackage{url} 
\usepackage{xspace,paralist}
\usepackage{lmodern}
\usepackage{float}
\floatstyle{boxed}
\restylefloat{figure}
\usepackage{multirow}
\usepackage{makecell}
\usepackage{ stmaryrd }

\usepackage{marvosym}
\usepackage{graphicx}
\usepackage{graphbox}
\usepackage{subcaption}
\usepackage{ bbold }
\usepackage{breakcites}
\usepackage{mathtools}

\usepackage{xfrac}

\usepackage{pifont}

\input{macros}

\floatstyle{plain}
\restylefloat{figure}

\begin{document}

\title{Toward Separating QMA from QCMA with a Classical Oracle}
\author{{\sc Mark Zhandry}\\{\tt mzhandry@gmail.com}\\NTT Research}
\date{}

\maketitle

\begin{abstract}QMA is the class of languages that can be decided by an efficient quantum verifier given a \emph{quantum} witness, whereas QCMA is the class of such languages where the efficient quantum verifier only is given a \emph{classical} witness. A challenging fundamental goal in quantum query complexity is to find a classical oracle separation for these classes. In this work, we offer a new approach towards proving such a separation that is qualitatively different than prior work, and show that our approach is sound assuming a natural statistical conjecture which may have other applications to quantum query complexity lower bounds.
\end{abstract}
	
\input{intro}
\input{prelim}

\input{conj}

\input{app-missing}

\bibliographystyle{alpha}
\bibliography{abbrev3,crypto,bib}

\end{document}

%% file: macros.tex
\theoremstyle{plain}

\newtheorem{theorem}{Theorem}
\newtheorem{lemma}[theorem]{Lemma} 
\newtheorem{corollary}[theorem]{Corollary}

\newtheorem{conjecture}[theorem]{Conjecture}

\newtheorem{definition}[theorem]{Definition}

\newtheorem{remark}[theorem]{Remark}

\numberwithin{theorem}{section}
\numberwithin{equation}{section}

\newcommand{\Z}{\mathbb{Z}}

\newcommand{\C}{\mathbb{C}}

\newcommand{\vv}{{\mathbf{v}}}

\newcommand{\xv}{{\mathbf{x}}}

\newcommand{\Id}{{\mathbf{I}}}

\newcommand{\Mm}{{\mathbf{M}}}

\newcommand{\Tr}{{\sf Tr}}

\newcommand{\As}{{\mathcal{A}}}

\newcommand{\Cs}{{\mathcal{C}}}
\newcommand{\Ds}{{\mathcal{D}}}

\newcommand{\Ms}{{\mathcal{M}}}
\newcommand{\Ns}{{\mathcal{N}}}
\newcommand{\Os}{{\mathcal{O}}}

\newcommand{\Qs}{{\mathcal{Q}}}

\newcommand{\Us}{{\mathcal{U}}}

\newcommand{\Xs}{{\mathcal{X}}}

\newcommand{\poly}{{\sf poly}}

\newcommand{\setupprotocol}[2]{%
	\expandafter\newcommand\csname protocol#1\endcsname{{\sf \Pi_{#2}}}%
	\expandafter\newcommand\csname gen#1\endcsname{{\sf Gen_{#2}}}%
	\expandafter\newcommand\csname genlossy#1\endcsname{{\sf GenLossy_{#2}}}%
	\expandafter\newcommand\csname setup#1\endcsname{{\sf Setup_{#2}}}%
	\expandafter\newcommand\csname enc#1\endcsname{{\sf Enc_{#2}}}%
	\expandafter\newcommand\csname dec#1\endcsname{{\sf Dec_{#2}}}%
	\expandafter\newcommand\csname extract#1\endcsname{{\sf Extract_{#2}}}%
	\expandafter\newcommand\csname derive#1\endcsname{{\sf Derive_{#2}}}%
	\expandafter\newcommand\csname embed#1\endcsname{{\sf Embed_{#2}}}%
	\expandafter\newcommand\csname rerand#1\endcsname{{\sf ReRand_{#2}}}%
	\expandafter\newcommand\csname trace#1\endcsname{{\sf Trace_{#2}}}%
	\expandafter\newcommand\csname goodtrace#1\endcsname{{\sf GoodTr_{#2}}}%
	\expandafter\newcommand\csname badtrace#1\endcsname{{\sf BadTr_{#2}}}%
	\expandafter\newcommand\csname gooddecoder#1\endcsname{{\sf GoodDec_{#2}}}%
	\expandafter\newcommand\csname goodf#1\endcsname{{\sf Goodfunc_{#2}}}%
	\expandafter\newcommand\csname badextract#1\endcsname{{\sf BadExtr_{#2}}}%
	\expandafter\newcommand\csname adv#1\endcsname{{\As_{#2}}}%	
	\expandafter\newcommand\csname sample#1\endcsname{{\sf Sample_{#2}}}%
	\expandafter\newcommand\csname diff#1\endcsname{{\sf Diff_{#2}}}%	
	\expandafter\newcommand\csname identify#1\endcsname{{\sf Indentify_{#2}}}%	
	\expandafter\newcommand\csname findtags#1\endcsname{{\sf FindTags_{#2}}}%	
	\expandafter\newcommand\csname confirmtags#1\endcsname{{\sf ConfirmTags_{#2}}}%	
	\expandafter\newcommand\csname eval#1\endcsname{{\sf Eval_{#2}}}%
	\expandafter\newcommand\csname sign#1\endcsname{{\sf Sign_{#2}}}%
	\expandafter\newcommand\csname ver#1\endcsname{{\sf Ver_{#2}}}%
	\expandafter\newcommand\csname prf#1\endcsname{{\sf F_{#2}}}%
	\expandafter\newcommand\csname Sim#1\endcsname{{\sf Sim_{#2}}}%

	\expandafter\newcommand\csname msk#1\endcsname{{\sf msk_{#2}}}%	
	\expandafter\newcommand\csname sk#1\endcsname{{\sf sk_{#2}}}%	
	\expandafter\newcommand\csname pk#1\endcsname{{\sf pk_{#2}}}%	
}

\setupprotocol{}{}

\newcommand{\Language}{{\sf L}}
\newcommand{\OILanguage}{{\sf OI\text{-} L}}
\newcommand{\QFT}{{\sf QFT}}
\newcommand{\QMA}{{\sf QMA}}
\newcommand{\QCMA}{{\sf QCMA}}
\newcommand{\OIQMA}{{\sf OI\text{-}QMA}}
\newcommand{\OIQCMA}{{\sf OI\text{-}QCMA}}
\newcommand{\NP}{{\sf NP}}

\newcommand{\Bernoulli}{{\sf Bernoulli}}

\def\E{\mathop{\mathbb E}}

\newcommand{\ignore}[1]{}

%% file: intro.tex
\section{Introduction}\label{sec:intro}

\emph{Do quantum witnesses offer more power than classical witnesses?} Slightly more precisely, there are two natural ways to generalize $\NP$ from the classical setting to the quantum setting: $\QMA$ (for Quantum Merlin-Arthur) is the set of languages decidable by efficient quantum algorithms with quantum witnesses, whereas $\QCMA$ (for Quantum \emph{Classical} Merlin-Arthur) is the set of languages decidable by efficient quantum algorithms with \emph{classical} witnesses. A long-standing fundamental question, first raised by~\cite{AhaNav02}, is whether or not these two generalizations of $\NP$ are the same.

As an unconditional separation between $\QMA$ and $\QCMA$ is out of reach given the state of complexity theory, the community has focused on proving \emph{oracle} separations. The first such separation~\cite{CCC:AarKup07} gives a \emph{quantum} oracle; that is, an oracle that implements a unitary operation and which can be queried on quantum states. A major open question has been whether there is a \emph{classical} oracle separation; that is, an oracle that implements a classical function, but is accessible in superposition. Classical oracle separations are considered more standard in the community. 

An early candidate classical oracle separation was given by~\cite{Lutomirski11}, but no proof was given. More recently, there have been several results making progress towards this goal by proving separations under different restrictions on how the oracle is accessed.~\cite{FefKim18} show a separation assuming the classical oracle is an ``in-place permutation oracle'', a non-standard modeling where the oracle irreversibly permutes the input state.~\cite{NatNir23} show a separation, assuming the witness is required to be independent of certain choices made in constructing the oracle. A very recent line of work has used quantum advantage relative to unstructured oracles~\cite{FOCS:YamZha22} to separate $\QMA$ from $\QCMA$:~\cite{EC:Liu23,ITCS:LLPY24} give a separation assuming the verifier can only make classical oracle queries, and more recently~\cite{BenKun24} give a separation which allows the verifier quantum queries, but assumes the adaptivity of the queries is sub-logarithmic. A standard classical oracle separation that makes no constraints on how the oracle is accessed or how the witness is created still remains open.

The central challenge in separating $\QMA$ from $\QCMA$ relative to a classical oracle seems to be the following. Consider a language in $\QMA\setminus\QCMA$, and consider measuring the $\QMA$ witness in the computational basis. The resulting classical string must \emph{not} be accepted by the $\QMA$ verifier, since otherwise we would then have a classical witness for the language, putting it in $\QCMA$. Therefore, in some sense, the $\QMA$ verifier needs to verify that the witness is in superposition. In order to allow for such verification using only a classical oracle, existing approaches require highly structured oracles. But then to actually prove the language is outside of $\QCMA$, we need a quantum query complexity lower-bound, and unfortunately the techniques we have are often not amenable to highly structured oracles. Additionally, any witness would naturally be treated as a form of oracle-dependent advice about the oracle, and most quantum query complexity techniques are not very good at distinguishing between quantum advice and classical advice. In other words, if a typical technique succeeded at proving a language is outside of $\QCMA$, then if it cannot distinguish quantum vs classical advice, it would likely also show that the language is outside $\QMA$ as well, thus failing to give a separation.

\paragraph{Our Work.} We give a new approach for separating $\QMA$ from $\QCMA$ relative to a classical oracle. Like prior work, we are unable to prove our oracle separation unconditionally. However, our separation appears much less structured than the prior work (though this is an intuitive statement rather than a formal one), and appears much more amenable to existing quantum query complexity techniques. In particular, we prove under a natural conjecture about $k$-wise independent distributions that our separation indeed works. We believe our work adds to the evidence that $\QMA$ and $\QCMA$ are indeed distinct relative to classical oracles, and may offer a new path toward a proof.

\subsection{Our Separating Oracles}

Our basic idea is the following. An instance in our language will correspond to a subset $S\subseteq [N]$ where $N$ is exponentially-sized. $S$ will be chosen to only contain a negligible fraction of $[N]$, but still be super-polynomial sized. We will provide a classical oracle (accessible in superposition) that decides membership in $S$. We will often simply call this oracle $S$.

Next, we choose a random state $|\psi\rangle=\sum_{y\in S}\alpha_y|y\rangle$ with support on the set $S$. The state $|\psi\rangle$ will be our witness state for the $\QCMA$ instance. An incomplete verification for $|\psi\rangle$ proceeds by simply checking that $|\psi\rangle$ has support contained in $S$. Essentially, $|\psi\rangle$ is acting as a $\QMA$ witness that $S$ is non-empty. But there is also a simple $\QCMA$ witness for this fact: any classical value $y\in S$.

We will instead attempt to turn $|\psi\rangle$ into a witness that $S$ is \emph{super-polynomial} sized. To do so, we will add a second oracle $U$ which essentially attests to $|\psi\rangle$ being a superposition over super-polynomially-many points. Intuitively, we want to show that $U$ can distinguish between $|\psi\rangle$ and any state whose support is only polynomial-sized.

To construct $U$, let $|\hat{\psi}\rangle$ denote the quantum Fourier transform (QFT) of $|\psi\rangle$. We observe that if $|\psi\rangle$ has support on a single point $y$, then $|\hat{\psi}\rangle$ will have uniform amplitude on all points in $[N]$ (though with complex phases on these points). On the other hand, for random $|\psi\rangle$ with support on the large set $S$, the amplitudes on different points will vary. Concretely, while the expected squared-amplitude on any point $z\in[N]$ is $1/N$, there is a reasonable chance that it could be, say, smaller than $1/2N$ or larger than $2/N$.

We will choose $U$ to be a subset of $[N]$ consisting of points where $|\hat{\psi}\rangle$ has squared-amplitude somewhat larger than $1/N$. We can then have, say, the total squared-amplitude of $|\hat{\psi}\rangle$ on points in $U$ be roughly $3/4$ while $|U|/N$ is only roughly $1/2$. In this case, the QFT of a classical string $y$ will have squared-amplitude on $U$ of only $1/2$. Thus, $U$ enables distinguishing $|\psi\rangle$ from a classical input. We will therefore give out an oracle for deciding membership in $U$. The verifier will first confirm that $|\psi\rangle$ has support only on $S$ using the oracle for $S$. Then it will compute $|\hat{\psi}\rangle$ via the QFT, and check that the support is in $U$. Overall the verifier accepts with probability $3/4$. Instances not in the language will consist of $S,U$ pairs where $S$ is very small but non-empty. We show that, for a certain way of choosing $U$, that there is no $\QMA$ witness in the case of such small $S$. Thus, our language is in $\QMA$ relative to the oracles for $S,U$.

\subsection{$\QCMA$ hardness} 

We now need a way to argue that our language does not have $\QCMA$ witnesses. While we showed that a classical string $y\in S$ cannot serve as a witness, this alone does not preclude some more clever way of attesting to $S$ being large. In particular, the witness $w$ could contain several points in $S$. Worse, perhaps queries to $U$ may reveal a significant amount of information about $S$, which may help deciding if $S$ is large or small. We make progress toward showing $\QCMA$ hardness of our oracle problem, as we now describe.

Consider a hypothetical $\QCMA$ verifier $V$ which is given a classical witness $w$ and makes quantum queries to $S,U$, and accepts in the case $S$ is large. We want to show that we can replace $S$ with a small set $S'$, and $V$ will still accept with too-high a probability, meaning it incorrectly claims that $(S',U)$ is in the language, despite $S'$ being small.

Toward that end, we will choose $S'$ to be all points in $S$ that are also ``heavy'' among the queries $V$ makes to $S$. That is, points $y\in S$ such that the query amplitude in $V$'s quantum queries to $S$ is above some inverse-polynomial threshold. As the total query amplitude of all points is just the number of queries of $V$ and hence polynomial, the number of heavy $y$ is polynomial. Hence, $S'$ is small. We can also construct $S'$ efficiently: for each heavy $y$, running $V$ and measuring a random query will have an inverse polynomial chance of producing $y$. By repeating this process a polynomial number of times, we can collect all heavy queries.

But why should $S$ and $S'$ be indistinguishable to $V$? By standard quantum query analysis, if $V$ can distinguish $S$ from $S'$, then it's queries must place significant amplitude on $S\setminus S'$. By measuring a random query, we therefore obtain with significant probability a $y\in S\setminus S'$. But since $y$ is not heavy, repeating this process many times will produce many different $y$. This means that if $V$ can distinguish $S$ from $S'$, it must actually be able to generate essentially arbitrarily large (polynomial) numbers of $y\in S$. Denote the number of $y$ by $L$.

\paragraph{$V$ that do not query $U$.} Let us first assume that $V$ makes no queries to $U$. In this case, we can argue that any distinguishing $V$ actually violates known query complexity results for multiple Grover search. In particular,~\cite{HamMag23} show that an algorithm making polynomially-many queries to a random sparse $S$ cannot produce $L$ points in $S$ except with probability bounded by $2^{-L}$ (see Lemma~\ref{lem:multi} for precise statement). Now, this result assumes no advice is provided about $S$, but the witness $w$ counts as advice. Fortunately, we can handle the advice using the strong exponential lower bound provided by~\cite{HamMag23}. Consider running the process above with a \emph{random} $w'$ instead of $w$. In the event $w'=w$, the process will produce $L$ points in $S$. Moreover, $w'=w$ with probability $2^{-|w|}$. By setting $L\gg |w|$ (recall that we can make $L$ an arbitrarily large polynomial), we therefore obtain an algorithm with no advice which produces $L$ points in $S$ with probability $2^{-|w|}\gg 2^{-L}$, contradicting the hardness of multiple Grover search.

\begin{remark}The above strategy inherently uses the fact that $w$ is classical. If $V$ had a quantum witness/advice, running it even once and measuring a random query to find a $y\in S$ would potentially destroy the witness, meaning further runs of $V$ are not guaranteed to produce \emph{any} points in $S$. This is the key place in the proof where we distinguish between classical and quantum witnesses, hopefully indicating a promising route toward proving a separation between $\QMA$ and $\QCMA$.
\end{remark}

\paragraph{$k$-wise independent $U$.} The above strategy does not work for handling queries to $U$. The problem is that $U$ takes potentially $N$ bits (which is exponential) to describe, meaning treating $U$ as advice would require setting $L\gg N$, at which point the bounds from~\cite{HamMag23} do not apply.

We will for now assume that $U$ is $k$-wise independent for a super-polynomial $k$, and return to justifying this assumption later.  We will assume such $k$-wise independence holds even conditioned on $S$ (but not conditioned on $|\psi\rangle$, whose correlation with $U$ is crucial for the correctness of our $\QMA$ verifier). A result of~\cite{C:Zhandry12} (formally described in Lemma~\ref{lem:kwiseindep}) shows that a $k$-wise independent $U$ is actually perfectly indistinguishable from a uniform $U$, for all quantum query algorithms making at most $k/2$ queries. Since $k$ is super-polynomial, we thus obtain perfect indistinguishability against all polynomial-query algorithms, including our process above for generating points in $S$. Consequently, the process above succeeds in generating $L$ points in $S$ even if $U$ is replaced by a uniformly random $U$ independent of $S$. But such a uniform $U$ can be simulated without knowledge of $S$ at all, and hence the lower-bound of~\cite{HamMag23} actually applies to algorithms making queries to uniform $U$. Thus under the assumption that $U$ is $k$-wise independent, we can justify $\QCMA$ hardness.

\paragraph{Our $U$ are ``close'' to $k$-wise independent.} We show that, by choosing $U$ carefully in a probabilistic way, $U$ is ``close'' to $k$-wise independent, even conditioned on $S$. By ``close'', we concretely mean that for every subset $T\subseteq[N]$ of size at most $k$, that $2^{-|T|}\leq \Pr[T\cap U=\emptyset]\leq (1+\epsilon)\times 2^{-|T|}$, for some very small $\epsilon$ which depends on $k,|S|,N$. Note that true $k$-wise independence is equivalent to $\Pr[T\cap U=\emptyset]=2^{-|T|}$ for all $T$ of size at most $k$.

\paragraph{Is this ``close'' enough?} Unfortunately, we do not know how to prove that $U$ which are close to $k$-wise independent are sufficient to make our approach work. The issue is that~\cite{C:Zhandry12} only applies to perfect $k$-wise independence, and there are counterexamples that show that the result does \emph{not} hold when replaced with some approximate notions of $k$-wise independence. 

The good news is that our notion of closeness is quite different from the usual notion of ``biased'' or ``almost'' $k$-wise independence used in the literature. Specifically, those notions allow for an additive error in any of the marginal probabilities, whereas we impose a strong \emph{multiplicative} error bound. This gives us hope that our distribution of $U$, despite not being truly $k$-wise independent, may still be close enough to get a separation.

We make progress toward justifying this claim. We make a conjecture that any distribution over $U$ which is ``close'' to $k$-wise independent (in our sense) can be turned into a distribution $U'$ that is truly $k$-wise independent. Importantly, $U$ and $U'$ will agree on almost all points. More precisely, for any $z$, the probability that $U$ and $U'$ differ on $z$ is negligible. See Conjecture~\ref{conj:stat} for the formal statement of this conjecture. Observe that this conjecture is simply a statement about distributions, and has nothing on the surface to do with quantum query complexity.

Under this conjecture, we complete the full oracle separation between $\QMA$ and $\QCMA$. Instead of giving out the oracle $U$, we simply give out the oracle $U'$, and prove $\QCMA$ hardness following the above approach. Our statistical conjecture is then used to show that replacing $U$ with $U'$ does not break our $\QMA$ verifier. Concretely, by standard query complexity arguments, we show that if our verifier works on $U$, then it must also work (with negligibly larger error) on $U'$.

\begin{remark}Our statistical conjecture gives \emph{one} way to prove an oracle separation between $\QMA$ and $\QCMA$ following our approach. Our conjecture could of course turn out to be false. Even in this case, our oracles still seem likely to give a separation, and there may be many other paths toward proving it. Perhaps if the general conjecture is false, our particular $U$ can still be converted into $U'$ as needed. Or maybe being ``close'' to $k$-wise independent is directly sufficient for a separation and can be proven via quantum query complexity arguments. Possibly there is a different distribution over $|\psi\rangle$ and associated $U$ where $U$ actually is perfectly $k$-wise independent. Thus, independent of our particular statistical conjecture, we believe our new approach at a separation gives a promising new approach toward separating $\QMA$ from $\QCMA$ relative to classical oracles.
\end{remark}

%% file: prelim.tex
\section{Preliminaries and notation}\label{sec:prelim}

A function $f(n)\leq n^{O(1)}$ is \emph{polynomial} and $f(n)\geq n^{-O(1)}$ is inverse polynomial. Functions $f(n)\geq n^{\omega(1)}$ are superpolynomial and $f(n)\leq n^{-\omega(1)}$ are negligible.

Let $\Bernoulli_p$ denote the distribution over $\{0,1\}$ which outputs $1$ with probability $p$. Let $\Bernoulli_p^N$ denote the distribution over $\{0,1\}^N$ consisting of $N$ iid samples from $\Bernoulli_p$. We will associate $\{0,1\}^N$ with subsets of $[N]$, where $S\subseteq[N]$ is associated with the vector $\vv$ such that $v_x=1$ if $x\in S$. We will also associate $\{0,1\}^N$ (and hence also subsets of $[N]$) with functions from $[N]$ to $\{0,1\}$, where $S$ is associated with its indicator function $f$, where $f(x)$ is 1 if and only if $x\in S$.

A joint distribution $X_1,\cdots X_n$ is $k$-wise independent if, for each subset $T$ of size at most $k$, $(X_i)_{i\in T}$ are independent random variables. $X_1,\cdots,X_n$ is $k$-wise \emph{uniform} independent if each $(X_i)_{i\in T}$ are independent uniform random variables over their respective domains.

\paragraph{Complex Normal Distribution.} Let $\Ns^\C_{\mu,\sigma}$ denote complex normal distribution with width $\sigma$. The probability density function of this distribution is given by $\Pr[x\gets\Ns^\C_{\mu,\sigma}]=\frac{1}{\pi\sigma^2}e^{-|x-\mu|^2/\sigma^2}$.

Two basic identities that will be useful when computing integrals involving complex normal distributions are the following. Let $\vv$ is a vector of $n$ complex numbers, $\Mm$ is a complex $n\times n$ matrix, and $\int_\C d\vv$ means to integrate over all complex vectors $\vv$. Then
\begin{align*}
	\int_\C e^{-\vv^\dagger\Mm\vv}d\vv&=\frac{\pi}{\det(\Mm)}&\int_\C |v_1|^2|v_2|^2 e^{-\vv^\dagger\Mm\vv}d\vv&=\frac{\pi^2 \Tr(\Mm)^2}{4\det(\Mm)^3}\text{ for }n=2
\end{align*}

\paragraph{Quantum Computation.} We assume the reader is familiar with the basics of quantum computation and quantum query models. Recall that in the standard model for making quantum queries to a classical oracle $\Os$, the algorithm submits a state $\sum_{x,y,z}\alpha_{x,y,z}|x,y,z\rangle$, and receives back $\sum_{x,y,z}\alpha_{x,y,z}|x,y\oplus \Os(x),z\rangle$. We will make use of the quantum Fourier tranform, denoted $\QFT_N$, defined on computational basis states as $|y\rangle\mapsto \frac{1}{\sqrt{N}}\sum_{z\in\Z_N}e^{i2\pi yz/N}|z\rangle$ for $y\in\Z_N$. We will usually drop the subscript $N$ as it will be clear from context. For a quantum state $|\psi\rangle$, we will typically let $|\hat{\psi}\rangle$ be shorthand for $\QFT|\psi\rangle$.

\subsection{Defining QMA and QCMA relative to Oracles}

We can consider two types of oracle versions of complexity classes, and in particular QMA/QCMA. The first, and most common, is to specify an infinite oracle $\Os:\{0,1\}^*\rightarrow\{0,1\}$, and define the classes QMA/QCMA relative to $\Os$. The second version, which is typically easier to work with, is to consider a variant of QMA/QCMA where the instance itself is an exponentially-sized oracle $\Xs:\{0,1\}^n\rightarrow\{0,1\}$. Fortunately, we show that a QMA/QCMA separation for one variant immediately gives such a separation for the other variant. Thus, it suffices to prove a separation for whichever variant is most convenient. 

\begin{definition}[Oracle-Aided QMA/QCMA]\label{def:qmaqcma} For a function $\Os:\{0,1\}^*\rightarrow\{0,1\}$, an oracle decision problem is a subset $\Language^\Os\subseteq\{0,1\}^*$ which depends on $\Os$. We say that $\Language^\Os$ is in $\QMA^\Os$ if there exists a polynomial-time oracle-aided quantum algorithm $V^\Os$ and polynomial $p$ such that:
	\begin{itemize}
		\item For every $x\in \Language^\Os$ of length $n$, there exists a state $|\psi\rangle$ on $p(n)$ qubits such that $\Pr[V^\Os(x,|\psi\rangle)=1]\geq 2/3$. In this case, we say that $V^\Os$ \emph{accepts} $x$.
		\item For every $x\notin \Language^\Os$ of length $n$, and for any state $|\psi\rangle$ on $p(n)$ qubits, $\Pr[V^\Os(x,|\psi\rangle)=1]\leq 1/3$. In this case, we say that $V^\Os$ \emph{rejects} $x$.
	\end{itemize}
	$\QCMA^\Os$ is defined analogously, except that the states $|\psi\rangle$ are required to be classical.
\end{definition}
\begin{definition}[Oracle-Input QMA/QCMA]\label{def:oiqmaqcma} Let $\Us=\{\Us_n\}_{n\in\Z^+}$ where $\Us_n$ are sets of strings $\Xs$ of length $2^{n^{\Theta(1)}}$, interpreted as functions from $\Xs:[n^{\Theta(1)}]\rightarrow\{0,1\}$. An oracle-input promise language is a subset $\OILanguage\subseteq \Us$. An oracle-input promise language $\OILanguage\subseteq \Us$ is in $\OIQMA$ if there exists a polynomial-time oracle-aided quantum algorithm $V$ and polynomial $p$ such that:
	\begin{itemize}
		\item For every $\Xs\in \OILanguage\cap \Us_n$, there exists a $|\psi\rangle$ on $p(n)$ qubits such that $\Pr[V^\Xs(|\psi\rangle)=1]\geq 2/3$. %In this case, we say that $V$ \emph{accepts} $\Xs$.
		\item For every $\Xs\in\Us_n\setminus\OILanguage$, and for any state $|\psi\rangle$ on $p(n)$ qubits, $\Pr[V^\Xs(|\psi\rangle)=1]\leq 1/3$. %In this case, we say that $V$ \emph{rejects} $\Xs$.
	\end{itemize}
	$\OIQCMA$ is defined analogously, except that the states $|\psi\rangle$ are required to be classical.
\end{definition}
Note that the constants $1/3,2/3$ in Definitions~\ref{def:qmaqcma} and~\ref{def:oiqmaqcma} are arbitrary, and can be replaced with $a,b$ for any $a,b$ such that $a\geq 2^{-\poly(n)},b\leq 1-2^{-\poly(n)},$ and $b-a\geq 1/\poly(n)$.

\medskip

Given any countable collection of oracles $\Os_1,\Os_2,\cdots$ which may have finite or infinite domains, we can construct a single oracle $\Os:\{0,1\}^*\rightarrow\{0,1\}$ by setting $\Os(j,x)=\Os_i(x)$, where $(j,x)$ is some encoding of pairs $(j,x)\in\{0,1\}^*\times\{0,1\}^*$ into strings in $\{0,1\}^*$. We can likewise convert $\Os$ back into $\Os_1,\Os_2,\cdots$. Therefore, we will take the definitions of $\QMA,\QCMA,\OIQMA,\OIQCMA$ to allow for verifiers making queries to countable collections of oracles.

This next theorem is proved in Section~\ref{sec:missproof} following a standard diagonalization argument, and shows that is suffices to give a separation for either variant of $\QMA,\QCMA$.
\begin{theorem}\label{thm:equiv}There exists a classical oracle $\Os$ such that $\QCMA^\Os\neq\QMA^\Os$ if and only if $\OIQCMA\neq\OIQMA$.
\end{theorem}
The style of oracle separation between $\QMA$ and $\QCMA$ given in~\cite{CCC:AarKup07} follows that of Definition~\ref{def:qmaqcma}, except that they use a quantum oracle instead of a classical oracle. They do not define the oracle-input versions of $\OIQMA$ and $\OIQCMA$ or a general result like Theorem~\ref{thm:equiv}, but their proof implicitly uses similar concepts. In particular, they first define a universe of quantum oracles: namely, those that are either the identity or reflect around a Haar random state. They then show that the two cases can be efficiently distinguished via a quantum witness, but not a classical witness. This is effectively showing a separation between $\OIQMA$ and $\OIQCMA$, except using quantum oracles instead of classical oracles. They then extend this to give a quantum oracle $\Qs$ separating $\QMA^\Qs$ from $\QCMA^\Qs$. This can be seen as roughly corresponding to a transformation like what we prove in Theorem~\ref{thm:equiv}. 

One slight non-triviality in generalizing their techniques to give a general equivalence is that a separation between $\OIQMA$ and $\OIQCMA$ only needs to show that, for any potential $\QCMA$ verifier, there is \emph{some} instance that the verifier answers incorrectly. This is indeed how our separation is constructed. In contrast,~\cite{CCC:AarKup07} show that almost all instances will fool any $\QCMA$ verifier. This stronger separation is crucially used when constructing the oracle $\Qs$ separating $\QMA^\Qs$ from $\QCMA^\Qs$, as they can basically choose $\Qs$ to be random from the appropriate universe of oracles. In order to get a general equivalence for arbitrary separations, including those like ours where the instance depends on the verifier, we have to work a bit harder and incorporate a diagonalization argument. Fortunately, this is standard.

\subsection{Some Useful Quantum Lemmas}

For an oracle algorithm $A$ making queries to an oracle $O$ and an oracle input $r$, let $\sum_{x,y,z}\alpha_{x,y,z}^{(i)}|x,y,z\rangle$ denote the $i$th oracle query, and let $M_x(i)=\sum_{y,z}|\alpha_{x,y,z}^{(i)}|^2$, which we will call the {\bf query mass} of $x$ in the $i$-th query.  Let $M_x=\sum_i M_x(i)$, the total query mass of $x$ over all $q$ queries, and for a subset $V$ let $M_V=\sum_{x\in V}M_x$ be the total query mass of points in $V$.

\begin{lemma}[\cite{BBBV97} Theorem 3.1+3.3]\label{lem:BBBV} Let $A$ be a quantum algorithm running making $q$ queries to an oracle $O$. Let $\epsilon > 0$ and let $V$ be a set of inputs. If we modify $O$ into an oracle $O'$ which is identical to $O$ except possibly on inputs in $V$, then $|\Pr[A^O()=1]-\Pr[A^{O'}()=1]|\leq 4\sqrt{q M_V}$.
\end{lemma}

\begin{lemma}[\cite{C:Zhandry12} Theorem 3.1]\label{lem:kwiseindep} Let $A$ be a quantum algorithm making $q$ queries to an oracle $O$. For any output $z$, the probability $A$ outputs $z$ when $O$ is $2q$-wise independent is identical to the probability $A$ outputs $z$ when $O$ is uniformly random.\end{lemma}

\begin{lemma}[\cite{HamMag23} Theorem 5.5]\label{lem:multi} Let $p\in [0, 1]$. There is a constant $C\leq 48e$ such that the following is true. The success probability of finding $K$ marked items in a random function $S : [N] \rightarrow \{0,1\}$ where $S(x) = 1$ with probability $p$ for each $x\in[N]$ is at most $(Cp(Q/K)^2)^K$ for any algorithm making $Q \geq K$ quantum queries to $S$.
\end{lemma}

%% file: conj.tex
\section{Our Oracle Separation}\label{sec:conj}

Here, we give our conjectured oracle separation between $\QMA$ and $\QCMA$. Following Theorem~\ref{thm:equiv}, it suffices to focus on the oracle-input variants of $\QMA$ and $\QCMA$. We first define our new statistical conjecture. Then we will define a certain ``nice'' type of distribution, which we call Fourier Independent (FI). We show that our statistical conjecture leads to the existence of such a FI distribution. Finally, we show that an FI distribution gives a separation between $\OIQMA$ and $\OIQCMA$, and hence an oracle $\Os$ such that $\QCMA^\Os\neq\QMA^\Os$ (via Theorem~\ref{thm:equiv}).

\subsection{Our Statistical Conjecture}

\begin{definition}[Substitution Distance] Consider two distributions $X_0,X_1$ over $\Sigma^\ell$ for some alphabet $\Sigma$. The \emph{substitution distance} between $X_0$ and $X_1$, denoted $\|X_0- X_1\|_{\sf sub}$, is the minimum value of $\epsilon\geq 0$ such that there is a joint distribution $(Z_0,Z_1)$ over $(\Sigma^\ell)^2$ where:
	\begin{itemize}
		\item The marginal distribution $Z_0$ is equal to $X_0$ and the marginal distribution $Z_1$ is equal to $X_1$.
		\item For each $i\in[\ell]$, $\Pr[Z_{0,i}\neq Z_{1,i}]\leq\epsilon$, where $Z_{b,i}$ means the $i$th entry of $Z_b$.
	\end{itemize}
\end{definition}
Intuitively, a small $\|X_0- X_1\|_{\sf sub}$ means that, by changing a few positions $X_0$ -- and each position with tiny probability -- we can turn $X_0$ into $X_1$.

\begin{conjecture}\label{conj:stat} There exist functions $r(N),k(N),\zeta(N),\eta(N)$ subject to the constraints $k(N)\geq (\log N)^{\omega(1)}$, $\eta(N)\leq (\log N)^{-\omega(1)}$, and $N\zeta(N)^3/r(N)^6\geq (\log N)^{\omega(1)}$  such that the following is true. Suppose $X_0$ is a distribution over $\{0,1\}^N$ such that for all sets $T\subseteq [N]$ of size of size at most $r(N)$, $2^{-|T|}\leq \Pr_{\xv\gets X_0}[x_i=0\forall i\in T]\leq (1+\zeta(N))\times 2^{-|T|}$. Then there exists a $k(N)$-wise uniform independent distribution $X_1$ such that $\|X_0-X_1\|_{\sf sub}\leq \eta(N)$.
\end{conjecture}

Think of $\log N$ as the instance size, so that $N$ is exponential. Conjecture~\ref{conj:stat} starts with a distribution $X_0$ that is in some sense very close to $r$-wise uniform independent, and concludes that $X_0$ must be negligibly-close in substitution error to an actual $k$-wise uniform independent distribution, for $k$ that may be different than $r$ but is still super-polynomial. Note that without the constraint $N\zeta(N)^3/r(N)^6\geq (\log N)^{\omega(1)}$, the conjecture is trivially true by setting $\zeta=0$ and $X_0=X_1$. The exact constraint we stipulate makes the conjecture non-trivial, and arises for technical reasons in our separation.

\subsection{Fourier Independent Distributions}

\begin{definition}\label{def:maindist} Let $N\in\Z^+$ and $S\subseteq [N]$. We say that a distribution $\Ds_S$ over pairs $(|\psi\rangle,U)$ is
$(k,\delta,\gamma)$-Fourier-Independent (FI) if the following hold:
\begin{itemize}
	\item $|\psi\rangle$ is a normalized superposition $\sum_{y\in S}\alpha_y|y\rangle$ over elements $y\in S$. 
	\item $U\subseteq[N]$, and the marginal distribution of $U$ is $k$-wise uniform independent.
	\item Let $|\hat{\psi}\rangle=\QFT|\psi\rangle$ be the quantum Fourier transform of $|\psi\rangle$. Then except with probability $\delta$ over the choice of $|\psi\rangle$, $\langle \hat{\psi}|\Pi_U |\hat{\psi}\rangle\geq \frac{1}{2}+\gamma$, where $\Pi_U$ is the projection operator $\sum_{z\in U}|z\rangle\langle z|$.

\end{itemize}
\end{definition}

We want $k$ to be large (say super-polynomial in $\log N$), $\delta$ to be small (say negligible in $\log N$), and $\gamma$ to be not-too-small (non-negligible in $\log N$). Then a Fourier Independent distribution is one where (1) $|\psi\rangle$ has support on $S$, (2) $U$ looks like a random set to query-bounded quantum algorithms (by Lemma~\ref{lem:kwiseindep}), but (3) $|\hat{\psi}\rangle$ is biased toward elements in $U$.

We now show the existence of a FI distribution for certain $k,\delta,\gamma$, assuming Conjecture~\ref{conj:stat}.

\begin{theorem}\label{thm:fi} Suppose Conjecture~\ref{conj:stat} holds. Then there exists functions $p,\epsilon,\delta,\gamma:\Z^+\rightarrow[0,1]$ and functions $N,k:\Z^+\rightarrow\Z^+$ such that (1) $p(n),\epsilon(n),\delta(n)\leq n^{-\omega(1)}$, (2) $N(n)\leq 2^{n^{O(1)}}$ (3) $k(n)\geq n^{\omega(1)}$, (4) $\gamma\geq n^{-O(1)}$, and (5) except with probability $\epsilon(n)$ over the choice of $S\gets\Bernoulli_{p(n)}^{N(n)}$, there exists a distribution $\Ds_S$ that is $(k,\delta,\gamma)$-Fourier-Independent.
\end{theorem}
\begin{proof}Let $N,\ell\in\Z^+$ be positive integers. For a subset $S\subseteq\Z_N$ of size $\ell$, let $\Ds_{S}'$ be the distribution over pairs $(|\psi\rangle,U')$ where $|\psi\rangle$ is a quantum state and $U'\subseteq\Z_N$, obtained as follows:
	\begin{itemize}
		\item For each $y\in S$, sample a complex number $\alpha_y\gets \Ns^\C(0,\sigma)$ where $\sigma=1/\sqrt{\ell}$. Let $|\psi\rangle=\sum_{y\in S}\alpha_y |y\rangle$
		\item For each $z\in \Z_N$, let $\beta_z=\sum_{y\in S}e^{i2\pi yz/N}\alpha_y$.
		\item For each $z\in\Z_N$, place $z$ into the output set $U'$ with independent probability $1-e^{-|\beta_z|^2}$.
	\end{itemize}
	
\paragraph{Intuition:} By choosing $\sigma=1/\sqrt{\ell}$, we will show that $|\psi\rangle$ is approximately normalized, so we can think of $|\psi\rangle$ as being essentially a random state with support on $S$. Let $|\hat{\psi}\rangle=\QFT|\psi\rangle$ be the QFT applied to $|\psi\rangle$. Then $|\hat{\psi}\rangle=\sum_{z\in\Z_N}(\beta_z/\sqrt{N})|z\rangle$. Thus, the set $U'$ will be biased toward containing the points $z$ where $\beta_z$ is large.  We will also show that $U'$ is close to $k$-wise independent in substitution distance, for an appropriate $k$. Via Conjecture~\ref{conj:stat}, this allows us to replace $U'$ with an appropriate $U$, giving a distribution $\Ds_S$ that is truly Fourier Independent.
	
We state some useful lemmas about $\Ds_S'$; the proofs will be given in Section~\ref{sec:missproof}. Define $\tau^S_T:=\Pr_{U'\gets\Ds_S}[T\cap U'=\emptyset]$, or equivalently, the probability $U'(z)=0$ for all $z\in T$, once $S$ is chosen.
	
\begin{lemma}\label{lem:detupperworst} For any subsets $S,T$, $1+|T|\leq (\tau^S_T)^{-1}\leq 2^{|T|}$\end{lemma}
	
\noindent Lemma~\ref{lem:detupperworst} holds for any $S,T$. In particular, for $|T|=1$, we have that $\tau^S_T=1/2$. This means each $z$ is placed in $U$ with probability $1/2$; by linearity of expectation, $\E[|U'|]=N/2$. We can also get a tighter lower-bound for larger $T$ with high probability over the choice of $S$:
	
\begin{lemma}\label{lem:detupperavg} For any $\epsilon>0$, except with probability at most $2N^2 e^{-\epsilon^2 \ell/2}$ over the choice of random subset $S$ of size $\ell$, for all subsets $T$, $(\tau^S_T)^{-1}\geq 2^{|T|}(1-|T|^2\epsilon)$. In this event, as long as $|T|^2\epsilon\leq 1/2$, we can bound $\tau^S_T\leq 2^{-|T|}(1+2|T|^2\epsilon)$.
\end{lemma}

\noindent Now we bound $\| \;|\psi\rangle\; \|^2$, showing that $|\psi\rangle$ is approximately normalized.
\begin{lemma}\label{lem:boundnorm}Fix any set $S$ of size $\ell$ and $\epsilon\in(0,1)$. Then except with probability $2e^{-\epsilon^2 \ell/8}$, $\|\;|\psi\rangle\; \|^2\in[1-\epsilon,1+\epsilon]$, where $|\psi\rangle$ is generated as in $\Ds'_S$.
\end{lemma}

\noindent Now, we bound the probability that (the normalization of) $|\hat{\psi}\rangle$ is accepted by $U'$. Note that $| \;|\psi\rangle\; |^2=| \;|\hat{\psi}\rangle\; |^2$.
\begin{lemma}\label{lem:probaccept} Except with probability at most $3(N^{-1}+\ell^2 N^{-2}) \epsilon^{-2}+4N^2 e^{-\ell\epsilon^2/32}$ over the choice of $S,U'$, we have $\left|\frac{\langle \hat{\psi}|\Pi_{U'} |\hat{\psi}\rangle}{\|\;|\psi\rangle\;\|^2}-3/4\right|\leq\epsilon$. Recall $\Pi_{U'}$ is the projection operator $\sum_{z\in U'}|z\rangle\langle z|$.
\end{lemma}
	
We now return to the proof of Theorem~\ref{thm:fi}. $\Ds_S'$ is almost Fourier Independent, except that the distributions over $U'$ are not $k$-wise uniform independent. However, Lemmas~\ref{lem:detupperworst} and~\ref{lem:detupperavg} show, in some sense, that $U'$ is very close to being $k$-wide uniform independent, for somewhat large $k$. We will then invoke the assumed Conjecture~\ref{conj:stat} to replace $U'$ with a distribution over $U$ that is actually $k$-wise uniform independent, for a sufficiently large $k$, and is close in substitution error to $U'$. We then show that the small substitution error between $U$ and $U'$ implies that the statements of Lemma~\ref{lem:probaccept} still approximately holds, even for $U$.
	
In more detail, let $X_0$ be the distribution over $U'$ stemming from $\Ds_S'$. Let $r(N),k(N),\zeta(N),\eta(N)$ be the functions guaranteed by Conjecture~\ref{conj:stat}. Define $N=\Theta(2^n)$ and $p = r(N)^4 \zeta(N)^{-2} N^{-1} \log^2 N$. By the conditions Conjecture~\ref{conj:stat} places on $r(N),k(N),\zeta(N),\eta(N)$, we therefore have that $p\leq n^{-\omega(1)} \log^2 N \leq n^{-\omega(1)}$. By standard concentration inequalities, the size $\ell$ of $S$ sampled from $\Bernoulli_p^N$ is very close to $p N=r(N)^4 \zeta(N)^{-2} \log^2 N$, except with negligible probability. Then set $\epsilon=\zeta(N)/r(N)^2$, which by the conditions of Conjecture~\ref{conj:stat} gives that $N^{-1}\epsilon^{-2}$, $\ell^2 N^{-2}\epsilon^{-2}$, and $N^2 e^{-\ell\epsilon^2/32}$ are all negligible. By Lemma~\ref{lem:detupperworst} and by invoking Lemma~\ref{lem:detupperavg} with $\epsilon=\zeta(N)/r(N)^2$, we have except with negligible probability over the choice of $S$, that $2^{-|T|}\leq \tau_S^T\leq 2^{-|T|}(1+\zeta)$ for all sets $T$ of size at most $r(N)$. Conjecture~\ref{conj:stat} then implies for ``good'' $S$ where the bounds on $\tau_S^T$ hold, there is a distribution $X_1$ that is $k(N)$-wise uniform independent and such that $\|X_0-X_1\|_{\sf sub}\leq \eta(N)$, with $k(N),1/\eta(N)\geq (\log N)^{\omega(1)}=n^{\omega(1)}$.
	
Let $S$ be a good set. Now consider the following distribution $\Ds_S$. It first uses the fact that $\|X_0-X_1\|_{\sf sub}\leq n^{-\omega(1)}$ to derive a joint distribution $(Z_0,Z_1)$ with the marginal $Z_0$ being equivalent to $X_0$ and $Z_1$ being $k$-wise uniform independent for $k\geq n^{\omega(1)}$. Then we sample $(|\psi\rangle,U')\gets\Ds_S'$ and sample $U$ from the distribution $Z_1$ conditioned on $Z_0=U'$. Let $|\psi'\rangle=|\psi\rangle/\|\;|\psi\rangle\;\|$. Output $(|\psi'\rangle, U)$. We then have that $U$ is distributed according to $Z_1$, and is therefore $k$-wise uniform independent. We also have by Lemma~\ref{lem:probaccept} that except with negligible probability, $\langle \hat{\psi}'|\Pi_{U'} |\hat{\psi}'\rangle\geq 2/3$. In order to justify Fourier Independence, we just need to show that, e.g. $\langle \hat{\psi}'|\Pi_{U} |\hat{\psi}'\rangle\geq 7/12$.
	
Toward that end, write $|\hat{\psi}'\rangle=\frac{1}{\sqrt{N}}\sum_{z\in\Z_N}\beta_z' |z\rangle$ where $\sum_z |\beta_z'|^2=1$. Then deinfe
\[E:=\langle \hat{\psi}'|\Pi_{U'} |\hat{\psi}'\rangle-\hat{\psi}'|\Pi_{U} |\hat{\psi}'\rangle=\sum_{z\in U'}|\beta_z'|^2 - \sum_{z\in U}|\beta_z'|^2 \leq \sum_{z\in U'\setminus U}|\beta_z'|^2\]
Now let $\xi_z$ denote the variable that is 1 if $z\notin U'\setminus U$ and is 0 otherwise. Then 	$E\leq \sum_z\xi_z |\beta_z'|^2$. Each $\xi_z$ is a 0/1 random variable with negligible expectation. Therefore, since $\sum_{z}|\beta_z'|^2=1$, we have that $\E[E]$ is negligible. Then by Markov's inequality, we have that $\Pr[E\geq \sqrt{\E[E]}]\leq \sqrt{\E[E]}$. Therefore, since $\E[E]$ and hence $\sqrt{\E[E]}$ are negligible, we have that $E$ is negligible except with negligible probability. This means in particular that $\langle \hat{\psi}'|\Pi_{U} |\hat{\psi}'\rangle\geq 7/12$, except with negligible probability. Thus $\Ds_S$ is $(k,\delta,\gamma)$-Fourier-Independent, thereby proving Theorem~\ref{thm:fi}.\end{proof}

\subsection{From FI Distributions to An Oracle Separation}

We now show that FI distributions as guaranteed by Theorem~\ref{thm:fi} give a separation between $\OIQCMA$ and $\OIQMA$. 

\begin{theorem}\label{thm:sep} Suppose Conjecture~\ref{conj:stat} holds. Then $\OIQCMA\neq\OIQMA$.
\end{theorem}
\begin{proof}We first invoke Theorem~\ref{thm:fi} to obtain distributions $\Ds_S$ satisfying properties (1) through (5) in the statement of Theorem~\ref{thm:fi}. We now use this to construct our separation. We will invoke the definition of $\OIQMA$ and $\OIQCMA$ with $a=1/2+\gamma(n)/2$ and $b=1/2+\gamma(n)$. Since $b-a=\gamma(n)/2$ is inverse polynomial, this is equivalent to the standard definition of $\OIQMA$ and $\OIQCMA$. Thus, a valid verifier will accept YES instances of size $n$ (given a correct witness) with probability at least $1/2+\gamma(n)$, and accept NO instances with probability at most $1/2+\gamma(n)/2$. 
	
An instance will be a pair of oracles $(S,U)$ where $S,U\subseteq [N(n)]$. Our verifier $V^{S,U}(|\psi\rangle)$ will do the following: $V$ will make a query to $S$ on the witness state $|\psi\rangle$ and measuring the output. If it accepts, then $V$ will apply the $\QFT$ to the witness state (which may have changed due to the measurement), and make a query to $U$, measuring the output. If both queries accept, then $V$ will output 1; if either measurement rejects, then $V$ will output 0.

We define the universe $\Us$ as the set of pairs $(S,U)$ for which either (1) there exists a state $|\psi\rangle$ such that  $\Pr[V^{S,U}(|\psi\rangle)=1]\geq 1/2+\gamma(n)$, or (2) for all states $|\psi\rangle$, $\Pr[V^{S,U}(|\psi\rangle)=1]< 1/2+\gamma(n)/2$. Then $\OILanguage\subseteq\Us$ is the set such that $\Pr[V^{S,U}(|\psi\rangle)=1]\geq 1/2+\gamma(n)$. By definition, $\OILanguage\in\OIQMA$.

We now show that $\OILanguage\notin \OIQCMA$. We can always assume without loss of generality that there is, say, a fixed quadratic running time $t$ such that the running time of any $\OIQCMA$ verifier is bounded by $t(|x|+|w|)$ where $|x|$ is the instance length and $|w|$ is the witness length. This is accomplished by padding the witness length with 0's that just get ignored by the verifier. 

Now let $q=q(n)$ be a sufficiently small super-polynomial function, which we will take as an upper bound on witness length. Let $Q(n)=q(n)^2$, which we take to be an upper bound on the number of queries when the witness length is at most $q(n)$. Let $v=v(n)$ be another super-polynomial. We will need the following constraints on $q,Q,v$:
\begin{align*}
	Q&\ll \gamma p^{-1/2}&
	v&\gg qQ^4\gamma^{-4}\\
	v&\ll \gamma N^{1/12}&
	Qk/v&\leq n^{-\omega(1)} 
\end{align*}
These can be satisfied for a sufficiently-small super-polynomial $q,v$ and for sufficiently large $n$. 

Suppose toward contradiction that there is such a $\OIQCMA$ verifier $V_*$. Sample $S\gets \Bernoulli_{p(n)}^{N(n)}$ and $(|\psi\rangle,U)\gets\Ds_S$, where $\Ds_S$ is $(k(n),\delta(n),\gamma(n))$-Fourier-Independent. Then we have that with overwhelming probability over the choice of $S,U,|\psi\rangle$, $\Pr[V^{S,U}(|\psi\rangle)=1]\geq 1/2+\gamma(n)$. Hence $(S,U)\in\OILanguage$. Therefore, under the assumption that $V_*$ is a verifier for $\OILanguage$, $V_*$ must accept $(S,U)$, meaning there is a classical witness $w$ such that $\Pr[V_*^{S,U}(w)=1]\geq 1/2+\gamma(n)$. 
	
We will now construct a different instance $S',U$, where $S'$ is constructed from the following algorithm ${\sf GenSmallSet}^{S,U}(w)$: initialize $S'=\{\}$, and then repeat the following loop for $i=1,...,v$:
	\begin{itemize}
		\item Run $V_*^{S,U}(w)$ until a randomly chosen query to $S$, and measure the query, obtaining a string $y_i\in[N]$.
		\item If $y_i\in S\setminus S'$, add $y_i$ to $S'$.
	\end{itemize}

\noindent The following is proved in Section~\ref{sec:missproof}:
\begin{lemma}\label{lem:pairwisesmall} If $U$ is sampled from a $k'$-wise uniform independent function and $S'$ is potentially correlated to $U$ but has size at most $v$, then except with probability $2N^2\left(\frac{\sqrt{ek'}}{\epsilon\sqrt{N}}\right)^{k'}$ over the choice of $U,S'$, it holds that, for any normalized state $|\phi\rangle$ with support on $S'$, $\langle\hat{\phi}|\Pi_U|\hat{\phi}\rangle\leq 1/2+v\epsilon$, where $|\hat{\phi}\rangle=\QFT|\phi\rangle$.
\end{lemma}
\begin{corollary}\label{cor:inL}Except with negligible probability over the choice of $S,U,S'$, $(S',U)\in\Us\setminus\OILanguage$.
\end{corollary}
\begin{proof}
Set $\epsilon=\gamma/3v$, and set $G:=2N^2\left(\frac{3v\sqrt{ek'}}{\gamma\sqrt{N}}\right)^{k'}$ for a yet-unspecified $k'$. By Lemma~\ref{lem:pairwisesmall}, except with probability at most $G$, $\Pr[V^{S',U}(|\phi\rangle)=1]\leq 1/2+\gamma/3< 1/2+\gamma/2$ for any state $|\phi\rangle$. In this case, $(S',U')\in \Us\setminus\OILanguage$. This event holds regardless of what $S'$ is, just using the fact that it consists of at most $v$ elements. Setting $k'=5$ and using that $v\leq \gamma N^{1/12}$ we have that $G=O(N^{1/12})$, which is negligible.
\end{proof}

\noindent We now show, however, that $V^*$ fails to reject $S',U$:
\begin{lemma}\label{lem:vstarbad} Except with negligible probability over the choice of $S,U,w,S'$, $\Pr[V_*^{S',U}(w)=1]> 1/2+\gamma(n)/2$
\end{lemma}
\begin{proof}Observe that the process ${\sf GenSmallSet}^{S,U}(w)$ for constructing $S'$ always generates $S'\subseteq S$. Now consider running $S'\gets{\sf GenSmallSet}^{S,U'}(w')$, where $w'$ is a random string and $U'$ is a random boolean function, with both $w',U'$ independent of $S$. This is an algorithm which makes $vQ$ queries to $S$ (and also to $U'$, but $U'$ us simulatable on its own since it is independent of $S$). $S$ in turn sampled from $\Bernoulli_{p}^N$. Finally, the algorithm outputs some subset $S'$ of $S$. By Lemma~\ref{lem:multi}, there is a universal constant $C$ such that 
	\[\Pr[|S'|\geq K:w',U'\text{ are uniform and independent of }S]\leq (Cp(vQ/K)^2)^K.\]
	
Now consider running $S'\gets {\sf GenSmallSet}^{S,U}(w')$, replacing $U'$ with $U$. Recall that $U$ is $k$-wise independent even conditioned on $S$, for $k\geq 2vQ$ where $vQ$ is the number of queries made by ${\sf GenSmallSet}$. By Lemma~\ref{lem:kwiseindep}, the output distribution of ${\sf GenSmallSet}^{S,U}(w')$ is identical to that of ${\sf GenSmallSet}^{S,U'}(w')$. Thus, we still have 
	\[\Pr[|S'|\geq K:U'=U, \text{ while }w'\text{ is uniform and independent of }S]\leq (Cp(vQ/K)^2)^K.\]

Finally, consider running $S'\gets{\sf GenSmallSet}^{S,U}(w)$. Since $\Pr[w'=w]=2^{-q}$, this means that with probability $2^{-q}$, ${\sf GenSmallSet}^{S,U}(w')$ is actually running ${\sf GenSmallSet}^{S,U}(w)$. Therefore, 
	\[\Pr[|S'|\geq K:U'=U,w'=w]\leq (Cp(vQ/K)^2)^K\times 2^q.\]

We now set $K=v Q\sqrt{2 C p} + q+\log_2(v)$. Then we have that $\Pr[|S'|\geq K:U'=U,w'=w]\leq 1/v$. Then in particular since $|S'|\leq v$ always, we have that $\E[|S'|]\leq (K-1)\Pr[|S'|<K]+v\Pr[|S'|\geq K]\leq (K-1)+1\leq K$.

Let $e_j$ be the probability that ${\sf GenSmallSet}$ adds an element $y_i$ to $S'$ in the $i$th iteration. Then $\sum_{j=1}^v e_j=\E[|S'|]\leq K$. We also have that the $e_j$ are monotonically decreasing since the $y_i$ are identically distributed and thus the probability mass outside of the growing $S'$ can only shrink. Thus $e_v\leq K/v$. 

For an input $y\in S$, let $M_y$ denote the total magnitude squared of $y$ in all queries made by $V_*^{S,U}(w)$ to oracle $S$. Then measuring a random choice of the $\leq Q$ queries by $V_*$, we will obtain $y$ with probability at least $M_y/Q$. Since $e_v\leq K/v$, this means that by the end of the loop, the expected total query weight of all points in $S$ but not in $S'$ is at most $QK/v$. By Markov's inequality, we therefore have that the query weight of points in $S\setminus S'$ is at most $\sqrt{QK/v}$, except with probability at most $\sqrt{QK/v}$. Recall that $\sqrt{QK/v}$ is negligible. Therefore, since this probability is negligible, we will therefore assume the total query weight of points in $S\setminus S'$ is at most $\sqrt{QK/v}$.

Lemma~\ref{lem:BBBV} shows that the difference in acceptance probability between $V_*^{S,U}(w)$ and $V_*^{S',U}(w))$ is at most $4\sqrt{Q\times \sqrt{QK/v}}=4\sqrt[4]{Q^3K/v}$. Observe that $Q^3K/v = Q^4\sqrt{2Cp}+Q^3(q+\log_2(v))/v$. Using that $Q\ll \gamma p^{-1/2}$ and $v\gg qQ^4/\gamma^4$, we therefore have that $Q^3K/v\leq \gamma^4/(12)^4$. Hence the difference in acceptance probability is at most $\gamma/3$, and hence $\Pr[V_*^{S',U}(w))=1]> 1/2+\gamma/2$.
\end{proof}

Corollary~\ref{cor:inL} and Lemma~\ref{lem:vstarbad} shows that $V_*$ fails to decide $\OILanguage$, contradicting it being a $\QCMA$ verifier. Hence, $\OILanguage\notin\QCMA$. This completes the proof of Theorem~\ref{thm:sep}\end{proof}

%% file: app-missing.tex
\section{Missing Proofs}\label{sec:missproof}

We now prove Lemmas~\ref{lem:detupperworst},~\ref{lem:detupperavg},~\ref{lem:boundnorm}, and~\ref{lem:probaccept} as well as Theorem~\ref{thm:equiv}. To do so, we introduce some extra notation. For a subset $S\subseteq\Z_N$, let $\Mm^S$ be the $N\times N$ matrix defined as $\QFT^\dagger\circ\Pi_S\circ\QFT$. Define $\tau^S_T:=\Pr_{U'\gets\Ds'_S}[T\cap U=\emptyset]$, or equivalently, the probability $U'(z)=0$ for all $z\in T$, once $S$ is chosen. Let $\Ms^S_T$ be the $|T|\times |T|$ sub-matrix of $\Mm^S$ consisting of rows and columns columns whose indices are in $T$. The following lemmas will be useful.

\begin{lemma}\label{lem:boundM} For any $\epsilon>0$, except with probability at most $2N^2 e^{-\epsilon^2 N^2/8\ell}$ over the choice of random $S$ of size $\ell$, $\max_{z\neq z'}|\Mm^S_{z,z'}|\leq \epsilon$
\end{lemma}
\begin{proof}Fix any $z\neq z'$, and consider the random variable $M^S_{z,z'}$ where $S$ is a random set of size $\ell$. We write:
	\[M^S_{z,z'}=\frac{1}{N}\sum_{y\in S}e^{i2\pi(z-z')y/N}=\frac{1}{N}\sum_{j=1}^\ell e^{i2\pi (z-z')y_j/N}\]
	where $y_1\cdots,y_\ell$ range over the elements of $S$, and are therefore random distinct values in $\Z_N$. 
	
	We first look at the real part of $\Mm^S_{z,z'}$, namely $\frac{1}{N}\sum_{i=1}^\ell Y_i$ where $Y_i=\cos(2\pi(z-z')y_i/N)$. For the moment, imagine the $y_i$ being uniform without the distinctness requirement, in which case since $z-z'\neq 0$ the $Y_i$ all become independent random variables bounded to the range $[-1,1]$. Moreover, the means are zero since $z-z'\neq 0\bmod N$. Then by Hoeffding's inequality, we have
	\[\Pr\left[\left|\sum_{i=1}^\ell Y_i\right|\geq t\right]\leq 2e^{-t^2/4\ell}\]
	We then observe that switching to distinct $y_i$ cannot make the inequality worse. This is because Hoeffding's inequality still holds when sampling is performed without replacement; in fact even better bounds are possible~\cite{Serfling74}.

	We now look at the imaginary part $\frac{1}{N}\sum_{i=1}^\ell Y_i'$ where $Y_i'=\sin(2\pi(z-z')y_i/N)$. By identical logic, we have that 
	\[\Pr\left[\left|\sum_{i=1}^\ell Y_i'\right|\geq t\right]\leq 2e^{-t^2/4\ell}\]
	Combining the two inequalities with the fact that the real and imaginary parts are orthogonal gives $\Pr[|M^S_{z,z'}|\geq \sqrt{2}t/N]\leq 4e^{-t^2/4\ell}$. Setting $t=\epsilon N/\sqrt{2}$ gives that 
	\[\Pr\left[\left|M^S_{z,z'}\right|\geq \epsilon\right]\leq 4e^{-\epsilon^2 N^2/8\ell}\]
	Taking a union bound over all $\binom{N}{2}$ off-diagonal terms\footnote{Since $M^S$ is Hermitian, we only need to bound, say, the terms above the diagonal.} we obtain the lemma.\end{proof}

\begin{lemma}\label{lem:marginals} Fix subsets $S,T$ with $|S|=\ell$. Then 
	\[\tau^S_T=\frac{1}{\det(\Id+\frac{N}{\ell}\Mm_S^T)}=\frac{1}{\det(\Id+\frac{N}{\ell}\Mm_T^S)}\]
\end{lemma}
\begin{proof}Recall that $U'$ is sampled by first sampling the state $|\psi\rangle$ with support on $S$ where each coefficient $\alpha_y$ is Gaussian distributed from $\Ns^\C(0,\sigma)$. Then each $z$ is excluded from $U'$ with probability $e^{-|\beta_z|^2}$, where $\beta_z/\sqrt{N}$ is the Fourier coefficient of $z$ in $|\psi\rangle$. Thus, the probability that $T\cap U'=\emptyset$, conditioned on $|\psi\rangle$ is just 
	\[\Pr[T\cap U'=\emptyset|\;|\psi\rangle]=e^{-\sum_{z\in T}|\beta_z|^2}=e^{-N\langle \psi|\QFT^\dagger \Pi_T \QFT|\psi\rangle}=e^{-N\langle \psi|\Mm^T|\psi\rangle}\]
	
	Now let $\vv$ be the vector of the $\alpha_y$ as $y$ varies in $S$. Then $|\psi\rangle$ is just $\vv$ with zeros inserted in each position outside of $S$. Then we can write \[\Pr[T\cap U'=\emptyset|\;|\psi\rangle]=e^{-N\vv^\dagger\Mm^T_S\vv}\]
	We now average over the choice of $\vv$, which is simply distributed as $\Ns^\C(0,\sigma)^\ell$. This gives us
	\begin{align*}\tau^S_T&=\int_{\C^{\ell}} \Pr[\vv]e^{-N\vv^\dagger\Mm^T_S\vv}d\vv\\
		&=\frac{1}{(\pi\sigma^2)^\ell}\int_{\C^{\ell}}e^{-\frac{|\vv|^2}{\sigma^2}-N\vv^\dagger\Mm^T_S\vv}d\vv\\
		&=\frac{1}{(\pi\sigma^2)^\ell}\int_{\C^{\ell}}e^{-\vv^\dagger(\sigma^{-2}\Id+N\Mm^T_S)\vv}d\vv\\
		&=\frac{1}{(\sigma^2)^\ell \det(\sigma^{-2}\Id+N\Mm^T_S)}=\frac{1}{\det(\Id+N\sigma^2 \Mm^T_S)}=\frac{1}{\det(\Id+\frac{N}{\ell} \Mm^T_S)}
	\end{align*}
	where in the last inequality we plugged in $\sigma=\sqrt{1/\ell}$. This gives us the first part of Lemma~\ref{lem:marginals}. Now, we observe that $\Mm^T_S=A^\dagger A$, where $A$ is the $|T|\times\ell$ sub-matrix of $\QFT$ restricted to column indices in $S$ and row indices in $T$. Then by the Weinstein–Aronszajn identity, $\det(\Id+\frac{N}{\ell} \Mm^T_S)=\det(\Id+\frac{N}{\ell}A^\dagger A)=\det(\Id+\frac{N}{\ell}AA^\dagger)=\det(\Id+\frac{N}{\ell}\Mm^S_T)$. This gives the second part of Lemma~\ref{lem:marginals}.\end{proof}

\noindent Going forward, we will use $\overline{\tau}^S_T$ as shorthand for $\det(\Id+\frac{N}{\ell}\Mm_T^S)$. Lemma~\ref{lem:marginals} says that $\tau^S_T=(\overline{\tau}^S_T)^{-1}$.

We can now prove Lemmas~\ref{lem:detupperworst},~\ref{lem:detupperavg},~\ref{lem:boundnorm}, and~\ref{lem:probaccept}.

\subsection{Proof of Lemma~\ref{lem:detupperworst}}

{
	\renewcommand{\thetheorem}{\ref{lem:detupperworst}}
	\begin{lemma} For any subsets $S,T$, $1+|T|\leq \overline{\tau}^S_T\leq 2^{|T|}$\end{lemma}
	\begin{proof}We first observe that the diagonal entries in $\Id+\frac{N}{\ell}\Mm^S$ are exactly $2$. Indeed, the $z$th diagonal entry is $1+(N/\ell)\left(\frac{1}{N}\sum_{x\in S}e^{i2\pi zx}e^{-i2\pi zx}\right)=2$. Moreover, $\frac{N}{\ell}\Mm^S$ and hence $\Id+\frac{N}{\ell}\Mm^S$ is PSD. Since $\Id+\frac{N}{\ell}\Mm_T^S$ is just a principle minor of $\Id+\frac{N}{\ell}\Mm^S$, it is also PSD with diagonal entries also equal to 2. The determinant is then bounded by the product of the diagonal entries, giving the upper bound.
		
		For the lower bound, we have that the eigenvalues of $\Id+\frac{N}{\ell}\Mm_T^S$ sum to $2|T|$, and each of the $|T|$ eigenvalues is at least 1. The determinant is the product of the eigenvalues, which is minimized when one of the eigenvalues is $|T|+1$ and the remaining $|T|-1$ eigenvalues are $1$. This gives the lower bound.
	\end{proof}
}

\subsection{Proof of Lemma~\ref{lem:detupperavg}}

{
	\renewcommand{\thetheorem}{\ref{lem:detupperavg}}
	\begin{lemma} For any $\epsilon>0$, except with probability at most $2N^2 e^{-\epsilon^2 \ell/2}$ over the choice of $S$, for all subsets $T$, $\overline{\tau}^S_T\geq 2^{|T|}(1-|T|^2\epsilon)$. In this event, as long as $|T|^2\epsilon\leq 1/2$, we can bound $\tau^S_T\leq 2^{-|T|}(1+2|T|^2\epsilon)$.
	\end{lemma}
	\begin{proof}Assume all the off-diagonal entries of $M^S$ are bounded by $2\epsilon \ell/N$, which by Lemma~\ref{lem:boundM} holds except with probability $2N^2 e^{-\epsilon^2 \ell/2}$. We have already seen that the diagonal entries of $\Id+\frac{N}{\ell}\Mm_T^S$ are exactly 2. The off-diagonal entries of $\Id+\frac{N}{\ell}\Mm_T^S$ are off-diagonal entries from $\Mm^S$ but scaled by $N/\ell$, and are therefore bounded by $2\epsilon$. By Gershgorin's circle theorem, the eigenvalues are therefore lower-bounded by $2(1-|T|\epsilon)$. The determinant is therefore lower-bounded by this quantity raised to the $|T|$. We obtain the lemma by bounding $(1-|T|\epsilon)^{|T|}\geq (1-|T|^2\epsilon)$.
	\end{proof}
}

\subsection{Proof of Lemma~\ref{lem:boundnorm}}

{
	\renewcommand{\thetheorem}{\ref{lem:boundnorm}}
	\begin{lemma}Fix any set $S$ of size $\ell$ and $\epsilon\in(0,1)$. Then except with probability $2e^{-\epsilon^2 \ell/8}$, $|\;|\psi\rangle\; |^2\in[1-\epsilon,1+\epsilon]$, where $|\psi\rangle$ is generated as in $\Ds'_S$.
	\end{lemma}
	\begin{proof}First, $\E[| \; |\psi\rangle \; |^2]=\sum_{y\in S}\E[|\alpha_y|^2]$.  Since $\alpha_y\gets\Ns^\Cs(0,1/\sqrt{\ell})$, the real and imaginary parts are mean-0 normal variables with variance $1/2\ell$. $|\alpha_y|^2$ is therefore distributed as the Chi-squared distribution with two degrees of freedom, scaled by $1/2\ell$, which therefore has expectation $1/\ell$. Summing over all $y\in S$ gives $\E[| \; |\psi\rangle \; |^2]=1$. We also see that $| \; |\psi\rangle \; |^2$ is distributed as a Chi-squared distribution with $2\ell$ degrees of freedom, scaled by $1/2\ell$. Via concentration inequalities for Chi-squared, we have that
		
		\[\Pr\left[\left|\;| \; |\psi\rangle \; |^2-1\;\right|\geq 4\sqrt{x/2\ell}\right]\leq 2e^{-x}\]
		Setting $\epsilon=4\sqrt{x/2\ell}$ gives the lemma.
	\end{proof}
}

\subsection{Proof of Lemma~\ref{lem:probaccept}}

{
	\renewcommand{\thetheorem}{\ref{lem:probaccept}}
	\begin{lemma} Except with probability at most $3(N^{-1}+\ell^2 N^{-2}) \epsilon^{-2}+4N^2 e^{-\ell\epsilon^2/32}$ over the choice of $S,U'$, we have $\left|\frac{\langle \hat{\psi}|\Pi_{U'} |\hat{\psi}\rangle}{|\;|\psi\rangle\;|^2}-3/4\right|\leq\epsilon$. Recall $\Pi_{U'}$ is the projection operator $\sum_{z\in U'}|z\rangle\langle z|$.
	\end{lemma}
	\begin{proof}Fix $|\psi\rangle$. We first compute the expectation of $\langle \hat{\psi}|\Pi_{U'} |\hat{\psi}\rangle$ (as $U$ varies) given $|\psi\rangle$. We will actually compute the complement $\langle \hat{\psi}|(\Id-\Pi_{U'}) |\hat{\psi}\rangle$
		\begin{align*}
			\E[\langle \hat{\psi}|(\Id-\Pi_{U'}) |\hat{\psi}\rangle]&=\frac{1}{N}\E[\sum_{z\notin U'}|\beta_z|^2]\\
			&=\frac{1}{N}\sum_{z\in\Z_N}\Pr[z\notin U']|\beta_z|^2\\
			&=\frac{1}{N}\sum_{z\in\Z_N}|\beta_z|^2e^{-|\beta_z|^2}\\
		\end{align*}
		Now we include the expectation as we vary $|\psi\rangle$, giving:
		\[\E[\langle \hat{\psi}|(\Id-\Pi_{U'}) |\hat{\psi}\rangle]=\frac{1}{N}\sum_{z\in\Z_N}\E[|\beta_z|^2e^{-|\beta_z|^2}]\]
		Next, we bound $\E[|\beta_z|^2e^{-|\beta_z|^2}]$. Since the distribution of $\alpha_y$ are invariant under phase change, we see that $\beta_z$ is just the sum of $\ell$ iid variables distributed as $\Ns^\Cs(0,1/\sqrt{\ell})$. This means each $\beta_z$ is distributed as $\Ns^\Cs(0,1)$. Breaking out the real and imaginary parts lets us write
		\begin{align*}
			\E[|\beta_z|^2e^{-|\beta_z|^2}]&=\int_\C |\beta|^2 e^{-|\beta|^2}\times \frac{1}{\pi}e^{-|\beta|^2}d\beta\\
			&=\frac{1}{\pi}\int_\C |\beta|^2 e^{-2|\beta|^2}d\beta=\frac{1}{\pi}\int_{-\infty}^\infty\int_{-\infty}^\infty (x^2+y^2)e^{-2(x^2+y^2)}dxdy=\frac{1}{4}
		\end{align*}
		Thus, we have that $\E[\langle \hat{\psi}|(\Id-\Pi_{U'}) |\hat{\psi}\rangle]=1/4$. Now we bound the variance. By carrying out a similar calculation as before, we have:
		\begin{align*}
			\E[\langle \hat{\psi}|(\Id-\Pi_{U'}) |\hat{\psi}\rangle^2]&=\frac{1}{N^2}\E[\left(\sum_{z\notin {U'}}|\beta_z|^2\right)]\\
			&=\frac{1}{N^2}\sum_{z,z'\in\Z_N}\Pr[z,z'\notin U']|\beta_z|^2|\beta_{z'}|^2\\
			&=\frac{1}{N^2}\left(\sum_z \Pr[z\notin U']|\beta_z|^4 + \sum_{z\neq z'}\Pr[z\notin U']\Pr[z'\in U']|\beta_z|^2|\beta_{z'}|^2\right)\\
			&=\frac{1}{N^2}\left(\sum_z |\beta_z|^4 e^{-|\beta_z|^2}+\sum_{z\neq z'}|\beta_z|^2|\beta_{z'}|^2 e^{-|\beta_z|^2-|\beta_{z'}|^2}\right)
		\end{align*}
		We include the expectation as we vary $|\psi\rangle$. The expectation of $|\beta_z|^4 e^{-|\beta_z|^2}$ becomes
		\begin{align*}
			\E\left[|\beta_z|^4 e^{-|\beta_z|^2}\right]&= \int_\C |\beta|^4 e^{-|\beta|^2}\times\frac{1}{\pi}e^{-|\beta|^2}=\frac{3}{4}
		\end{align*}
		Meanwhile, to compute the expectation of $|\beta_z|^2|\beta_{z'}|^2 e^{-|\beta_z|^2-|\beta_{z'}|^2}$, we use that the co-variance matrix of $\beta_z,\beta_{z'}$ is exactly $\Mm^S_{\{z,z'\}}$. Thus
		\begin{align*}
			&\E\left[|\beta_z|^2|\beta_{z'}|^2 e^{-|\beta_z|^2-|\beta_{z'}|^2}\right]\\
			&\;\;\;\;\;\;=\int_\C\int_\C |\beta|^2|\beta'|^2 e^{-|\beta|^2-|\beta'|^2}\times \frac{1}{\pi^2\det(\Mm^S_{\{z,z'\}})^2}e^{-\left(\begin{array}{cc}\beta^*&(\beta')^*\end{array}\right)\cdot\left(\Mm^S_{\{z,z'\}}\right)^{-1}\cdot \left(\begin{array}{c}\beta\\\beta'\end{array}\right)}d\beta d\beta'\\
			&\;\;\;\;\;\;=\frac{1}{\pi^2\det(\Mm^S_{\{z,z'\}})^2}\int_\C\int_\C |\beta|^2|\beta'|^2 e^{-\left(\begin{array}{cc}\beta^*&(\beta')^*\end{array}\right)\cdot\left[\Id+\left(\Mm^S_{\{z,z'\}}\right)^{-1}\right]\cdot \left(\begin{array}{c}\beta\\\beta'\end{array}\right)}d\beta d\beta'\\
			&\;\;\;\;\;\;=\frac{1}{\pi^2\det(\Mm^S_{\{z,z'\}})^2}\times \frac{\pi^2\Tr(\Id+(\Mm^S_{\{z,z'\}})^{-1})^2}{4\det(\Id+(\Mm^S_{\{z,z'\}})^{-1})^3}\\
			&\;\;\;\;\;\;=\frac{\det(\Mm^S_{\{z,z'\}})\Tr(\Id+(\Mm^S_{\{z,z'\}})^{-1})^2}{4\det(\Id+\Ms^S_{\{z,z'\}})^3}
		\end{align*}
		Now write $\Mm^S_{\{z,z'\}}=\left(\begin{array}{cc}1&\epsilon_{z,z'}\\\epsilon_{z,z'}&1\end{array}\right)$. By Lemma~\ref{lem:boundM}, we can now bound $|\epsilon_{z,z'}|$ by $\ell/4N<1/2$, except with probability $2N^2e^{-\ell/32}$. Then we have that 
		\[\E[|\beta_z|^2|\beta_{z'}|^2 e^{-|\beta_z|^2-|\beta_{z'}|^2}]=\frac{\left(2-|\epsilon_{z,z'}|^2\right)^2}{\left(1-|\epsilon_{z,z'}|^2\right)\left(4-|\epsilon_{z,z'}|^2\right)^3}\leq \frac{1}{16}(1+|\epsilon_{z,z'}|^2)\]
		where the last inequality used that $|\epsilon_{z,z'}|\leq 1/2$.
		
		We therefore have that $\E[\langle \hat{\psi}|(\Id-\Pi_{U'}) |\hat{\psi}\rangle^2]\leq \frac{3}{4N}+\frac{1}{16}+\frac{1}{16N^2}\sum_{z\neq z'}|\epsilon_{z,z'}|$. Since the mean is $1/4$, the variance is therefore at most $\frac{3}{4N}+\frac{1}{16N^2}\sum_{z\neq z'}|\epsilon_{z,z'}|^2\leq \frac{3}{4N}+\frac{1}{16N^2}\sum_{z\neq z'}\left(\frac{\ell}{4N}\right)^2\leq \frac{3}{4N}+\frac{\ell^2}{256N^2}$. We now apply Chebyshev's inequality, to get that 
		\[\Pr[|\langle \hat{\psi}|(\Id-\Pi_{U'}) |\hat{\psi}\rangle-1/4|\geq \epsilon/2]\leq \frac{\frac{3}{4N}+\frac{\ell^2}{64N^2}}{(\epsilon/2)^2}+2N^2e^{-\ell/32} \]
		
		Now we have that
		\begin{align*}\Pr\left[\left|\frac{\langle \hat{\psi}|\Pi_{U'} |\hat{\psi}\rangle}{|\;|\psi\rangle\;|^2}-3/4\right|\geq\epsilon\right]
			&=\Pr\left[\left|\frac{\langle \hat{\psi}|(\Id-\Pi_{U'}) |\hat{\psi}\rangle}{|\;|\psi\rangle\;|^2}-1/4\right|\geq\epsilon\right]\\
			&=\Pr\left[\langle \hat{\psi}|(\Id-\Pi_{U'}) |\hat{\psi}\rangle\notin |\;|\psi\rangle\;|^2\times\left[\frac{1}{4}-\epsilon,\frac{1}{4}+\epsilon\right]\right]\\
			&\leq \Pr\left[\langle \hat{\psi}|(\Id-\Pi_{U'}) |\hat{\psi}\rangle\notin \left[\frac{1}{4}-\epsilon/2,\frac{1}{4}+\epsilon/2\right]\right]\\
			&\;\;\;\;\;\;\;\;\;+\Pr[|\;|\psi\rangle\;|^2\notin [1-\epsilon/2,1+\epsilon/2]]\\
			&\leq \left(\frac{\frac{3}{4N}+\frac{\ell^2}{64N^2}}{(\epsilon/2)^2}+2N^2e^{-\ell/32} \right)+ 2e^{-\ell\epsilon^2/32}\\
			&\leq 3(N^{-1}+\ell^2 N^{-2}) \epsilon^{-2}+4N^2 e^{-\ell\epsilon^2/32}
		\end{align*}
	\end{proof}
}

\subsection{Proof of Lemma~\ref{lem:pairwisesmall}}

{
	\renewcommand{\thetheorem}{\ref{lem:pairwisesmall}}
	\begin{lemma}If $U'$ is sampled from a $k'$-wise uniform independent function and $S'$ is potentially correlated to $U'$ but has size at most $v$, then except with probability $2N^2\left(\frac{\sqrt{ek'}}{\epsilon\sqrt{N}}\right)^{k'}$ over the choice of $U',S'$, it holds that, for any normalized state $|\phi\rangle$ with support on $S'$, $\langle\hat{\phi}|\Pi_{U'}|\hat{\phi}\rangle\leq 1/2+v\epsilon$, where $|\hat{\phi}\rangle=\QFT|\phi\rangle$.
	\end{lemma}
}
\begin{proof}For a set $U'$, let $\Mm^{U'}=\QFT^\dagger \Pi_{U'} \QFT$. Observe that $(\Mm^{U'})_{z,z'}=\frac{1}{N}\sum_{y\in \Z_N}e^{i2\pi (z-z')y/N}\xi_y$, where $\xi_y$ is the boolean value that is 1 if $y\in U'$. Then we have that $(\Mm^{U'})_{z,z}=|U'|/N$. In expectation, $|U'|/N$ is $1/2$. We now bound how far $|U'|/N$ can deviate from $1/2$. Recall that $|U'|/N=\frac{1}{N}\sum_{y\in\Z_N}\xi_y$. The $\xi_y$ are not truly independent so we cannot use standard concentration inequalities such as Hoeffding. However, we can use the following somewhat standard bound (e.g.~\cite{Tao10}) for the $[-1,1]$-weighted sum of $k'$-wise independent boolean random variables:
	\[\Pr\left[\left|\frac{1}{N}\sum_{y\in \Z_N}\xi_y-\frac{1}{2}\right|\geq \epsilon/\sqrt{2}\right]\leq 2\left(\frac{\sqrt{ek'}}{\epsilon\sqrt{N}}\right)^{k'}\]

	On the other hand, for $z'\neq z'$, taking the expectation over $U'$, we see that $\E_{U'}[(\Mm^{U'})_{z,z}]=0$. We now try to bound the deviation from the expectation, by bounding the real and imaginary parts separately. We see that $\Re[(\Mm^{U'})_{z,z'}]=\frac{1}{N}\sum_{y\in \Z_N}\cos(2\pi (z-z')y/N)\xi_y$, which is the weighted sum of boolean random variables $\xi_y$, where the weights are all in $[-1,1]$. Using the $[-1,1]$-weighted sum of $k'$-wise independent boolean random variables again, we have:
	\[\Pr[\frac{1}{N}\sum_{y\in \Z_N}\cos(2\pi (z-z')y/N)\xi_y\geq \epsilon/\sqrt{2}]\leq 2\left(\frac{\sqrt{ek'}}{\epsilon\sqrt{N}}\right)^{k'}\]
	Combining with an identical bound on the imaginary part of $(\Mm^U)_{z,z'}$, we have that
	\[\Pr[|(\Mm^{U'})_{z,z'}|\geq\epsilon]\leq 4\left(\frac{\sqrt{ek'}}{\epsilon\sqrt{N}}\right)^{k'}\]
	By union-bounding over all entries of $\Mm^U$ (which only needs to count entries on the diagonal and above since $\Mm$ is Hermitian), we have except for probability at most $2N^2\left(\frac{\sqrt{ek'}}{\epsilon\sqrt{N}}\right)^{k'}$, both
	(1) for all $z$ that $(\Mm^U)_{z,z}\in[1/2-\epsilon,1/2+\epsilon]$ and (2) for all $z\neq z'$ that $|(\Mm^U)_{z,z'}|\leq \epsilon$.

	Now suppose (1) and (2) hold. Now consider a supposed set $S'$ of size $v$. Let $\Mm^{U'}_{S'}$ be the $v\times v$ sub-matrix whose row and column indices are in $S'$. Then by the Gershgorin circle theorem, all eigenvalues of $\Mm^{U'}_{S'}$ are bounded from above by $1/2+v\epsilon$.
\end{proof}

\subsection{Proof of Theorem~\ref{thm:equiv}}

{
	\renewcommand{\thetheorem}{\ref{thm:equiv}}
	\begin{theorem}There exists a classical oracle $\Os$ such that $\QCMA^\Os\neq\QMA^\Os$ if and only if $\OIQCMA\neq\OIQMA$
	\end{theorem}
	\begin{proof}In one direction, assume a classical oracle $\Os$ such that $\QCMA^\Os\neq\QMA^\Os$. Let $\Language^\Os$ be the separating language, and $V$ a verifier for $\Language^\Os$.
		
	We construct a language $\OILanguage$ as follows. Let $Q(n)$ be a (polynomial) upper bound on the length of queries that $V$ makes to $\Os$ when given an instance $x$ of length $n$. Let $\Os_n$ be the portion of $\Os$ corresponding to queries of length at most $Q(n)$. Then $|\Os_n|\leq 2^{n^{\Theta(1)}}$.
	
	Let $\Us_n$ consist of strings of the form $(x,\Os_n)$ where $x$ has size $n$, and let $\OILanguage$ be the set of $(x,\Os_n)$ where $x\in\Language^\Os$. We can decide membership in $\OILanguage$ as follows: given a witness state $|\psi\rangle$, first recover $x$ by making $n$ queries to $(x,\Os_n)$. Then run $V^{\Os_n}(x,|\psi\rangle)$. By our choice of $\Os_n$ containing all of $\Os$ that gets queried by $V$, we therefore have that $V^{\Os_n}(x,|\psi\rangle)$ is identical to $V^\Os(x,|\psi\rangle)$. As such, if $V$ correctly decides if $x\in\Language^\Os$, then our verifier correctly decides in $(x,\Os_n)\in\OILanguage$. Thus, $\OILanguage\in\OIQMA$.
	
	On the other hand, suppose there is a $\OIQCMA$ verifier $V_*$ that decides $\OILanguage$. We can then readily obtain a $\QCMA$ verifier for $\Language^\Os$. On input instance $x$ and witness $w$, simulate $V_*^{(x,\Os_n)}(w)$ by answering queries to $x$ with the bits of $x$, and queries to $\Os_n$ by forwarding the queries to $\Os$. If $V_*$ decides membership in $\OILanguage$, this exactly decides if $x\in\Language^\Os$.
	
	\medskip
	
	We now turn to the other direction in the statement of Theorem~\ref{thm:equiv}. Assume that $\OIQCMA\neq\OIQMA$. Let $\Us$ be the universe and $\OILanguage\subseteq\Us$ be the separating language, and $V$ the $\OIQMA$ verifier for $\OILanguage$. 
	
	We construct our oracle $\Os$ and associated language $L^\Os$ as follows. $\Os$ will be interpreted as a countably-infinite family of oracles $(\Xs_{n_i})_{i\in\Z^+}$ for integers $n_i\in\Z^+$. Let $\Os_j=(\Xs_{n_i})_{i\leq j}$. By padding the witness with 0's size appropriately, we can assume that any potential $\QCMA$ verifier runs in quadratic time. Let $T_1,T_2,\cdots$ be an enumeration over all oracle-aided quadratic-time quantum algorithms. $\Os$ and $\Language^\Os$ will be constructed by constructing $\Xs_{n_i}$ for $i=1,2,\cdots$ as follows: Consider the $\OIQCMA$ verifier $V_i^\Xs(w)$ which has $\Os_{i-1}$ hard-coded, and runs $T_i^{\Os_{i-1},\Xs}(w)$. Since $\Os_{i-1}$ is constant-sized, $V_i^\Xs(w)$ runs in quadratic time (though with a large constant overhead coming from $\Os_{i-1}$). Let $n_i$ be an integer and $\Xs_{n_i}\in\Us_{n_i}$ be an instance such that both:
		\begin{itemize}
			\item $V_i$ incorrectly decides $\Xs_{n_i}$ given classical witnesses/inputs. That is, either $\Xs_{n_i}\in\OILanguage$ and $V_i^{\Xs_{n_i}}(w)$ rejects for all $w$, or $\Xs_{n_i}\notin\OILanguage$ and there exists a $w$ such that $V_i^{\Xs_{n_i}}(w)$ accepts.
			\item $n_i$ is large enough so that the inputs to the function $\Xs_{n_i}$ are longer than the running time of $T_{j}^{\Os_{j}}$ on inputs of length $n_{j}$ for all $j<i$, meaning $T_{j}^{\Os_{j}}$, and hence $V_{j}$, on inputs of length $n_j,j<i$ never query any input that is an input to $\Xs_{n_i}$.
		\end{itemize}
		
	Such an $\Xs_{n_i}$ is guaranteed to exist: that \emph{some} $n_i,\Xs_{n_i}$ exist satisfying the first criteria follows immediately from the fact that $\OILanguage\notin\OIQCMA$. Suppose that there is no $n_i,\Xs_{n_i}$ satisfying the second criteria. This means there are only a finite number of $\Xs$ where $V_i$ fails to correctly decide. But by hard-coding these bad instances along with the correct solutions, we can obtain a new verifier $V_i'$ which correctly decides $\OILanguage$, contradicting that $\OILanguage\notin\QCMA$.
		
	We then let $\Language^\Os$ consist of the strings $0^{n_i}$ such that $\Xs_{n_i}\in \OILanguage$.  We immediately see that $\Language^\Os\in\QMA^\Os$: by making appropriate queries to $\Os$, we can simulate the $\OIQMA$ verifier $V^{\Xs_{n_i}}$, thereby deciding membership of $0^{n_i}$ in $\OILanguage$.
	
	We now show that $\Language^\Os\notin\QCMA^\Os$. Consider a supposed $\QCMA$ verifier $V_*$. This corresponds to some $T_i$ according to our enumeration. Then we argue that $V_*$ incorrectly decides membership for $0^{n_i}$. Indeed, we know that $T_i^\Os$ in instance $0^{n_i}$ never queries on $\Xs_{n_j}$ for $j>i$, by our choice of $n_j$. Hence, it can be simulated as $T_i^{\Os_i}$, or equivalently $T_i^{\Os_{i-1},\Xs_{n_i}}$. Thus, $V_*$ corresponds exactly to the verifier $V_i$. But we know that $V_i$ incorrectly decides membership for $\Xs_{n_i}$. Hence, $V_*$ is not a valid $\QCMA$ verifier.\end{proof}
}